\documentclass[copyright,creativecommons]{eptcs}
\providecommand{\event}{GandALF 2023} %

\newif\ifextended
\extendedfalse

\usepackage{iftex}

\ifpdf
  \usepackage{underscore}         %
  \usepackage[T1]{fontenc}        %
\else
  \usepackage{breakurl}           %
\fi

\usepackage{amsmath,latexsym,amssymb,amsfonts,amscd,amsthm,stmaryrd,textcomp}
\usepackage{commath}
\usepackage[noend,linesnumbered,ruled,vlined,resetcount]{algorithm2e}
\usepackage[]{todonotes}
\usepackage{bm}
\usepackage{blkarray, bigstrut}
\usepackage{cleveref}
\Crefname{example}{Example}{Examples}
\usepackage{mfirstuc}
\usepackage{tikz, pgfplots}
\usetikzlibrary{intersections}
\usetikzlibrary{calc,fit,tikzmark}
\usetikzlibrary{%
	decorations.pathreplacing,%
	decorations.pathmorphing,%
	decorations.markings
}
\usepackage{rwthcolors}
\usepackage{subcaption}
\usepackage{cutwin}

\newcommand{\chngd}[1]{\setlength{\fboxrule}{1.5pt}\fcolorbox{black!25}{black!25}{#1}}
\newcommand{\chng}[1]{\setlength{\fboxrule}{1.5pt}\fcolorbox{black!25}{white}{#1}}

\newcommand{\boxchngd}[1]{\tikz[remember picture,overlay]{\node[yshift=5pt,fill=black!25,fit={(pic cs:A#1)($(pic cs:B#1)+(.75\linewidth,.8\baselineskip)$)}] {};}}
\newcommand{\boxchngdstart}[1]{\tikzmark{A#1}}
\newcommand{\boxchngdend}[1]{\tikzmark{B#1}}

\newcommand{\refalgobase}{Algorithm~\hyperref[algo:fmplex1]{1a} }
\newcommand{\refalgobounds}{Algorithm~\hyperref[algo:leaving-out-bounds]{1b} }
\newcommand{\refalgobacktracking}{Algorithm~\hyperref[algo:backtracking]{1c} }

\newtheorem{theorem}{Theorem}
\newtheorem{definition}{Definition}
\newtheorem{lemma}{Lemma}
\newtheorem{example}{Example}

\DeclareMathAlphabet{\mathcal}{OMS}{cmsy}{m}{n} %

\title{FMplex: A Novel Method for Solving\\ Linear Real Arithmetic Problems}
\author{Jasper Nalbach\thanks{Jasper Nalbach was supported by the DFG RTG 2236 UnRAVeL.}
\institute{RWTH Aachen University, Germany}
\email{nalbach@cs.rwth-aachen.de}
\and
Valentin Promies 
\institute{RWTH Aachen University, Germany}
\email{promies@cs.rwth-aachen.de}
\and
Erika Ábrahám 
\institute{RWTH Aachen University, Germany}
\email{abraham@cs.rwth-aachen.de}
\and
Paul Kobialka
\institute{University of Oslo, Norway}
\email{paulkob@ifi.uio.no}
}
\def\titlerunning{FMplex: A Novel Method for Solving Linear Real Arithmetic Problems}
\def\authorrunning{J. Nalbach, V. Promies, E. Ábrahám and P. Kobialka}

\newcommand{\vect}[1]{\bm{#1}}
\newcommand{\minus}{\text{-}}
\newcommand{\rvect}[2]{\vect{#1_{#2,{\minus}}}}
\newcommand{\cvect}[2]{\vect{#1_{{\minus},#2}}}
\newcommand{\matr}[1]{\capitalisewords{#1}}
\newcommand{\sys}[2]{\matr{#1}\vect{x} \leq \vect{#2}}
\newcommand{\submatr}[2]{\matr{#1}\text{\smaller{$[#2]$}}}
\newcommand{\subvect}[2]{\vect{#1}\text{\smaller{$[#2]$}}}
\newcommand{\subsys}[3]{\submatr{#1}{#3}\vect{x} \leq \subvect{#2}{#3}}
\newcommand{\subsyseq}[3]{\submatr{#1}{#3}\vect{x} = \subvect{#2}{#3}}

\newcommand{\bnd}[2]{\textit{bnd}_{#1}(#2)}

\newcommand{\lbs}[2]{I_{#2}^{-}(\matr{#1})}
\newcommand{\ubs}[2]{I_{#2}^{+}(\matr{#1})}
\newcommand{\nbs}[2]{I_{#2}^{0}(\matr{#1})}
\newcommand{\sbs}[2]{I_{#2}^{*}(\matr{#1})}

\newcommand{\nats}{\mathbb{N}}
\newcommand{\reals}{\mathbb{R}}
\newcommand{\rationals}{\mathbb{Q}}

\newcommand{\solset}[1]{\textit{sol}(#1)}

\newcommand{\rank}[1]{\textit{rank}(#1)}

\newcommand{\fmpproj}[3]{P_{#2,#3}(#1)}

\newcommand{\proj}[2]{{#1}|_{#2}}

\newcommand{\sat}[0]{\textit{SAT}}
\newcommand{\unsat}[0]{\textit{UNSAT}}
\newcommand{\partialunsat}[0]{\textit{PARTIAL-UNSAT}}

\newcommand\fmplexelim[0]{\textnormal{\texttt{FMP}}}
\newcommand\fmelim[0]{\textnormal{\texttt{FM}}}

\newcommand{\dom}{\textit{dom}}

\newcommand{\vs}{/\!\!/}

\newif\ifproofs\proofstrue

\begin{document}
\maketitle

\begin{abstract}
In this paper we introduce a novel quantifier elimination method for conjunctions of  \emph{linear real arithmetic} constraints. Our algorithm is based on the \emph{Fourier-Motzkin variable elimination} procedure, but by case splitting we are able to reduce the worst-case complexity from doubly to singly exponential. The adaption of the procedure for SMT solving has strong correspondence to the \emph{simplex algorithm}, therefore we name it \emph{FMplex}. Besides the theoretical foundations, we provide an experimental evaluation in the context of SMT solving.
\end{abstract}
    
\section{Introduction}
\emph{Linear real arithmetic (LRA)} is a powerful first-order theory with strong practical relevance. We focus on checking the satisfiability of \emph{conjunctions} of LRA constraints, which is needed e.g. for solving quantifier-free LRA formulas using \emph{satisfiability modulo theories (SMT) solvers}.
The problem is known to be solvable in \emph{polynomial} worst-case complexity but, surprisingly, the \emph{ellipsoid} method \cite{KHACHIYAN198053} proposed in 1980 by Khachiyan is still the only available algorithm that implements this bound. However, this method is seldomly used in practice due to its high average-case effort. Instead, most approaches employ the \emph{simplex} algorithm introduced by Dantzig in 1947, which has a \emph{singly exponential} worst case complexity, but which is quite efficient in practice.
A third available solution is the \emph{Fourier-Motzkin variable elimination (FM)} method, proposed in 1827 by Fourier \cite{fourier1827analyse} and re-discovered in 1936 by Motzkin \cite{motzkin1936beitrage}. In contrast to the other two approaches, FM admits quantifier elimination, but it has a \emph{doubly exponential} worst case complexity, even though there have been various efforts to improve its efficiency by recognizing and avoiding redundant computations (e.g. \cite{imbert1993fourier,JingComplexityFME}). 

In this paper, we propose a novel method, which is derived from the FM method, but which turns out to have striking resemblance to the simplex algorithm. This yields interesting theoretical insights into the relation of the two established methods and the nature of the problem itself. More precisely, our contributions include:
\begin{itemize}
    \item The presentation of \emph{FMplex}, a new variable elimination method based on a divide-and-conquer approach. We show that it does not contain certain redundancies Fourier-Motzkin might generate and it lowers the overall complexity from \emph{doubly} to \emph{singly} exponential.
    \item An adaptation of FMplex for SMT solving, including methods to prune the search tree based on structural observations.
    \item A theorem formalizing connections between FMplex and the simplex algorithm.
    \item An implementation of the SMT adaptation and its experimental evaluation.
\end{itemize}
After recalling necessary preliminaries in \Cref{sec:preliminaries}, we introduce our novel FMplex method first for quantifier elimination in \Cref{sec:varelim} and then for SMT solving in \Cref{sec:sat}.
We present related work and compare FMplex with other methods, first qualitatively in \Cref{sec:othermethods}, and then experimentally in \Cref{sec:experiments}. We discuss future work and conclude the paper in \Cref{sec:conclusion}.

\ifextended
\else

An extended version of this paper including more detailed proofs can be found on arXiv \cite{nalbach2023fmplex}. 
\fi

\section{Preliminaries}
\label{sec:preliminaries}

Let $\reals$, $\rationals$ and $\nats$ denote the set of real, rational respectively natural ($0\notin\nats$) numbers.
For $k \in \nats$ we define $[k]:=\{1,\ldots,k\}$. Throughout this paper, we fix $n \in \nats$, a set $X = \{ x_1,\ldots,x_n \}$ and a corresponding vector $\vect{x} = (x_1,\ldots,x_n)^T$ of $\reals$-valued variables.

\vspace{-0.55em}\paragraph{Matrices}
For $m \in \nats$ let $\matr{e}^{(m)} \in \rationals^{m \times m}$ be the identity matrix, and $\vect{0}^{(m)} = (0\ \cdots\ 0)^T \in\rationals^{m\times 1}$.
The $i$th component of $\vect{f}\in \rationals^{m\times 1}\cup\rationals^{1\times m}$ is denoted by $f_i$ and the component-wise comparison to zero by $\vect{f} \geq 0$.
For $\matr{a} \in \rationals^{m \times n}$, $\rvect{a}{i} \in \rationals^{1 \times n}$ and $\cvect{a}{i} \in \rationals^{m \times 1}$ denote the $i$th row respectively column vector of $\matr{a}$.
Furthermore, $\submatr{a}{I}$ denotes the sub-matrix of $\matr{a}$ containing only the rows with indices from some $I\subseteq [m]$.
For $\vect{f} \in \rationals^{1\times m}$, $\vect{f}\matr{a}$ is a \emph{linear combination} of the rows $i \in [m]$ of $\matr{a}$ with $f_i \neq 0$.
We call $\matr{a}$ \emph{linearly independent} if none of its rows is a linear combination of its other rows, and \emph{linearly dependent} otherwise.
The \emph{rank of \matr{a}} $\rank{\matr{a}}$ is the size of a maximal $I \subseteq [m]$ with $\submatr{a}{I}$ linearly independent.

\vspace{-0.55em}\paragraph{Linear Constraints}
Let $\vect{a} = (a_1, \ldots, a_n) \in \rationals^{1\times n}$, $b \in \rationals$ and $\sim \in \{ =,\leq,<,\neq \}$ a \emph{relation symbol}.
We call $\vect{a}\vect{x}$ a \emph{linear term} and $\vect{a}\vect{x} \sim b$ a \emph{linear constraint}, which is \emph{weak} if $\sim \in \{=,\leq\}$  and \emph{strict} otherwise.
A \emph{system of linear constraints}, or short a \emph{system}, is a non-empty finite set of linear constraints.
For most of this paper, we only consider constraints of the form $\vect{a}\vect{x} \leq b$.
We can write every system $C=\{\rvect{a}{i}\;\vect{x} \leq b_i \mid i \in [m]\}$ of such constraints in \emph{matrix representation} $\sys{a}{b}$ with suitable $\matr{A} \in \rationals^{m \times n}$ and $\vect{b} \in \rationals^{m\times 1}$.
Conversely, every row $\rvect{a}{i}\; \vect{x} \leq b_i, \ i \in [m]$ of $\sys{A}{b}$ is a linear constraint.
Thus, the representations are mostly interchangeable; however, the matrix representation allows redundant rows in contrast to the set notation.
As the latter will play a role later on, we will stick to the matrix representation.

\vspace{-0.55em}\paragraph{Variable Assignment}
An \emph{assignment} is a function $\alpha: Y \to \reals$ with domain $\dom(\alpha)=Y\subseteq X$.
The \emph{extension} $\alpha[x_i \mapsto r]$ is the assignment with domain $\dom(\alpha)\cup\{x_i\}$ such that $\alpha[x_i \mapsto r](x_j)=\alpha(x_j)$ for all $x_j\in \dom(\alpha)\setminus\{x_i\}$ and $\alpha[x_i \mapsto r](x_i)=r$.
For $Z\subseteq Y$, the \emph{restriction} $\proj{\alpha}{Z}$ is the assignment with domain $Z$ such that $\proj{\alpha}{Z}(x_i)=\alpha(x_i)$ for all $x_i\in Z$.
We extend these notations to sets of assignments accordingly.

The standard \emph{evaluation} of a linear term $t$ under $\alpha$ is written $\alpha(t)$.
We say that $\alpha$ \emph{satisfies} (or is a solution of) a constraint $c = (\vect{a}\vect{x} \sim b)$ if $\alpha(a_1 x_1 + \ldots a_n x_n) \sim b$ holds, and denote this fact by $\alpha \models c$.
All solutions of $c$ build its \emph{solution set} $\solset{c}$.
Similarly, $\alpha\models (\matr{a} \vect{x} \leq \vect{b})$ denotes that $\alpha$ is a common solution of all linear constraints in the system $\matr{a} \vect{x} \leq \vect{b}$.
A system is \emph{satisfiable} if it has a common solution, and \emph{unsatisfiable} otherwise.
 Note that each satisfiable system has also a rational-valued solution.

\vspace*{0.4em}
\noindent We will also make use of the following two well-known results.

\begin{theorem}[Farkas' Lemma \cite{farkas1902theorie}] \label{thm:farkas}
    Let $\matr{a} \in \rationals^{m \times n}$ and $\vect{b} \in \rationals^{m\times 1}$.
	Then the system $\matr{a} \vect{x} \leq \vect{b}$ is satisfiable if and only if for all $\vect{f} \in \rationals^{1\times m}$ with $\vect{f} \geq 0$ and $\vect{f} \matr{a} = (0,\ldots,0) \in \rationals^{1\times n}$ it holds $\vect{f} \vect{b} \geq 0$.
\end{theorem}

\begin{theorem}[Fundamental Theorem of Linear Programming, as in \cite{luenberger1984linear}]\label{thm:linearprogramming}
	Let $\matr{a} \in \rationals^{m \times n}$ and $\vect{b} \in \rationals^{m \times 1}$.
	Then $\sys{a}{b}$ is satisfiable if and only if there exists a subset $I \subseteq [m]$ such that $\submatr{a}{I}$ is linearly independent, $\abs{I} = \rank{\matr{a}}$, and there exists an assignment $\alpha: X \to \reals$ with $\alpha \models (\submatr{a}{I}\vect{x} = \subvect{b}{I})$ and $\alpha \models (\sys{a}{b})$.
\end{theorem}

\subsection{Fourier-Motzkin Variable Elimination}
The \emph{Fourier-Motzkin variable elimination} (FM) \cite{fourier1827analyse,motzkin1936beitrage} method allows to eliminate any $x_j \in X$ from a system $\sys{a}{b}$ by computing $\sys{a'}{b'}$ with $\cvect{a'}{j}=0$ and such that an assignment $\alpha$ is a solution of $\sys{a'}{b'}$ if and only if there is $r\in\rationals$ so that $\alpha[x_j\mapsto r]$ is a solution of $\sys{a}{b}$.
Graphically, the solution set of $\sys{a'}{b'}$ is the projection of the solutions of $\sys{a}{b}$ onto $X \setminus \{ x_j \}$.

The idea of the FM method is as follows.
For each $i \in [m]$ with $a_{i,j} \neq 0$, the constraint $\rvect{a}{i}\;\vect{x} \leq b_i$ can be rewritten as either a \emph{lower bound} or an \emph{upper bound} on $x_j$, denoted in both cases as $\bnd{j}{\rvect{a}{i}\;\vect{x} \leq b_i}$:
\[ \big(\sum_{k \in [n] \setminus \{j\}} -\frac{a_{i,k}}{a_{i,j}} \cdot x_k\big) +\frac{b_i}{a_{i,j}} \leq x_j,\ \, \text{ if } a_{i,j}<0,\qquad \text{ resp. }\qquad x_j \leq \big(\sum_{k \in [n] \setminus \{j\}} -\frac{a_{i,k}}{a_{i,j}} \cdot x_k\big) +\frac{b_i}{a_{i,j}},\ \,\text{ if $a_{i,j}>0$}.\]

\begin{definition}
	For $\matr{a} \in \rationals^{m \times n}$, we define the index sets
    \[\lbs{a}{j} := \{i \in [m] \mid a_{i,j} < 0\},\quad
    \ubs{a}{j} := \{i \in [m] \mid a_{i,j} > 0\},\quad and \quad \nbs{a}{j} := \{i \in [m] \mid a_{i,j} = 0\}.\]
\end{definition}
\noindent $\lbs{a}{j}$, $\ubs{a}{j}$ and$\nbs{a}{j}$ indicate the rows of $\sys{a}{b}$ which induce lower bounds, upper bounds and no bounds on $x_j$, respectively.
Due to the density of the reals, there exists a value for $x_j$ that satisfies all bounds if and only if each lower bound is less than or equal to each upper bound.
However, since in general the involved bounds are symbolic and thus their values depend on the values of other variables, we cannot directly check this condition.
To express this, we let $\sys{a'}{b'}$ be defined by the constraint set
\[
    \{\bnd{j}{\rvect{a}{\ell}\;\vect{x} \leq b_{\ell}} \leq \bnd{j}{\rvect{a}{u}\;\vect{x} \leq b_{u}} \mid (\ell,u) \in \lbs{a}{j}\times\ubs{a}{j}\}\quad \cup\quad \{\rvect{a}{i}\;\vect{x} \leq b_i \mid i \in \nbs{a}{j}\}.
\]

\noindent In matrix representation, the FM method applies the following transformation:

\begin{definition}[Fourier-Motzkin Variable Elimination]
	Let $\matr{a} \in \rationals^{m \times n}$, $\vect{b} \in \rationals^{m\times 1}$, and $j \in [n]$.
	Let further
    $m'=|\lbs{a}{j}|\cdot |\ubs{a}{j}|+|\nbs{a}{j}|$ and $\matr{f} \in \rationals^{m' \times m}$ be a matrix consisting of exactly the following rows:\footnote{Remember that we use lower case letters for rows of matrices with the respective upper case letter as name. Thus, $\rvect{e^{(m)}}{i}$ denotes the $i$th column vector of the identity matrix $\matr{E^{(m)}}$.}
    \begin{align*}
        -\frac{1}{a_{\ell,j}} \cdot \rvect{e^{(m)}}{\ell} + \frac{1}{a_{u,j}} \cdot \rvect{e^{(m)}}{u}\ \text{ for every pair }\ (\ell,u) \in \lbs{a}{j}\times\ubs{a}{j}\ \qquad \text{ and }\ \qquad 
        \rvect{e^{(m)}}{i}\ \text{ for every }\ i \in \nbs{a}{j}.
    \end{align*}

	\noindent Then the \emph{Fourier-Motzkin variable elimination} $\fmelim_{j}(\sys{a}{b})$ of $x_j$ from the system $\sys{a}{b}$ is defined as the system $\matr{f}\matr{a}\vect{x}\leq \matr{f} \vect{b}$.
\end{definition}

The consistency of $\sys{a}{b}$ can be checked by successively eliminating variables $x_n, \ldots, x_1$, obtaining intermediate systems $\sys{a^{(n-1)}}{b^{(n-1)}}, \ldots, \sys{a^{(0)}}{b^{(0)}}$.
All entries of the transformation matrix $\matr{f}$ in the definition above are positive, and thus for any $k \in \{0,\ldots,n-1\}$ and any row $i'$ in $\sys{a^{(k)}}{b^{(k)}}$,
there exists $0\leq \vect{f} \in \rationals^{m\times 1}$ s.t. $\vect{f}\matr{A} = \rvect{a^{(k)}}{i'}$ and $\vect{f}\vect{b} = b^{(k)}_{i'}$, or in short:
$\sum_{i \in [m]} f_{i} \cdot (\rvect{a}{i}\;\vect{x} \leq b_i) = (\rvect{a^{(k)}}{i'}\vect{x} \leq b^{(k)}_{i'})$.
We call this kind of linear combinations \emph{conical combinations}.
By Farkas' Lemma (\Cref{thm:farkas}), if $\sys{a^{(0)}}{b^{(0)}}$ is unsatisfiable, then so is $\sys{a}{b}$.
If it is satisfiable, then it is satisfied by the empty assignment, which can be extended successively to a model of $\sys{a^{(1)}}{b^{(1)}},\ldots, \sys{a^{(n-1)}}{b^{(n-1)}}$ and $\sys{a}{b}$.

A major drawback of the Fourier-Motzkin variable elimination is its doubly exponential complexity in time and space w.r.t. the number of eliminated variables.
Moreover, many of the generated rows are redundant because they are linear combinations of the other rows, i.e. they could be omitted without changing the solution set of the system.
Redundancies might already be contained in the input system, or they arise during the projection operation.
While removing all redundancies is expensive, there are efficient methods for removing some redundancies of the latter type, for example Imbert's acceleration theorems \cite{imbert1990redundant,imbert1993fourier,JingComplexityFME}.

\begin{lemma}[Redundancy by Construction]\label{def:redundancyconstruction}
	Let $\matr{A} \in \rationals^{m \times n}, \vect{b} \in \rationals^{m\times 1}$ and $\matr{F} \in \rationals^{m' \times m}$.
Let furthermore $\matr{A'}=\matr{F}\matr{A}$, $\vect{b'}=\matr{F}\vect{b}$ and $i \in [m']$.
    If there exists $\vect{r} \in \rationals^{1\times m'}$ with $\vect{r} \geq 0$, $r_i = 0$ and $\vect{r}\matr{F} = \rvect{f}{i}$ (i.e. the $i$th row of $\sys{a'}{b'}$ is a conical combination $\vect{r}\matr{F}\matr{A}\vect{x} \leq \vect{r}\matr{F}\vect{b}$ of the other rows), then that row is redundant in $\sys{a'}{b'}$, i.e. the solution set does not change when omitting it: $\solset{\sys{a'}{b'}} = \solset{\subsys{a'}{b'}{[m'] \setminus \{ i \}}}$.
\end{lemma}

\section{FMplex as Variable Elimination Procedure}
\label{sec:varelim}

The FM method encodes that none of the lower bounds on some variable $x_j$ in a system $\sys{a}{b}$ is larger than any of its upper bounds.
In our \emph{FMplex} method, instead of considering all lower-upper bound combinations at once, we \emph{split the problem into a set of sub-problems} by case distinction either on \emph{which of the lower bounds is the largest} or alternatively on \emph{which of the upper bounds is the smallest}.
For splitting on lower bounds, for each lower bound on $x_j$ we consider solutions where this lower bound is maximal under all lower bounds, and at the same time not larger than any of the upper bounds.
The upper bound case is analogous.
Then $\sys{a}{b}$ is satisfiable if and only if there exists a solution in one of these sub-problems.
Asymptotically, these sub-problems are significantly smaller than the systems produced by FM, so that in total our approach produces \emph{at most exponentially} many constraints after iterated application, in contrast to the doubly exponential effort of the FM method.

Formally, if there are no upper or no lower bounds on $x_j$, then there is no need for case  splitting and we follow FM using $\exists x_j.\; \sys{a}{b} \equiv \subsys{a}{b}{\nbs{a}{j}}$.
Otherwise, for the sub-problem when designating $i \in \lbs{a}{j}$ as largest lower bound, we encode that no other lower bound is larger than the bound induced by row $i$, and no upper bound is below this bound.
Using the set notation for systems, we obtain
\begin{align*}
	&\{ \bnd{j}{\rvect{a}{i'}\;\vect{x} \leq b_{i'}} \leq \bnd{j}{\rvect{a}{i}\;\vect{x} \leq b_i} \mid i' \in \lbs{a}{j},\; i' \neq i\} \\
	\cup & \{\bnd{j}{\rvect{a}{i}\;\vect{x} \leq b_i} \leq \bnd{j}{\rvect{a}{i'}\;\vect{x} \leq b_{i'}} \mid i' \in \ubs{a}{j}\} %
	\cup  \{\rvect{a}{i'}\;\vect{x} \leq b_{i'} \mid i' \in \nbs{a}{j}\}.
\end{align*}

\begin{example} \label{example:fmplex-idea}
	We eliminate $x_2$ from the system $\sys{a}{b}$ consisting of the lower-bounding constraints $c_1$ and $c_2$, and the upper-bounding $c_3$ and $c_4$, specified below along with a graphical depiction.
	The lower bounds $\lbs{a}{2} = \{1,2\}$ on $x_2$ are blue, the upper bounds $\ubs{a}{2} = \{3,4\}$ are green. The solution set is the gray area and the dashed line indicates the split into two sub-problems, namely the cases that $c_1$ resp. $c_2$ is a largest lower bound on $x_2$ and not larger than any upper bound on $x_2$.

\begin{figure}[h]
	\begin{subfigure}[h]{0.48\textwidth}
		\centering
		\begin{align*}
			\begin{blockarray}{ccc}
			\begin{block}{c[cc]}
				c_1 & -1 & -1 \bigstrut[t] \\
				c_2 & 0 & -2 \\
				c_3 & -2 & 1 \\
				c_4 & 0 & 1 \bigstrut[b]\\
			\end{block}
			\end{blockarray}
			\cdot
			\left[ {\begin{array}{ccc}
				x_1 \\ x_2
			\end{array} } \right]
			\leq
			\left[ {\begin{array}{ccc}
				-4 \\ -2 \\ 1 \\ 5
			\end{array} } \right]
		\end{align*}
	\end{subfigure}\hfill
	\begin{subfigure}[h]{0.48\textwidth}
		\centering
		\begin{tikzpicture}[scale=0.5,
		lb/.style={
			very thick, color=blue100, opacity=1, shorten <= -5pt, shorten >= -5pt,
			postaction={draw,decorate,decoration={markings,mark=at position .4 with {\draw[->, shorten <= 0, shorten >= 0] (0,0) -- (0,-8pt);}}}}, 
		ub/.style={
			very thick, color=green100, opacity=1, shorten <= -5pt, shorten >= -5pt,
			postaction={draw,decorate,decoration={markings,mark=at position .6 with {\draw[->, shorten <= 0, shorten >= 0] (0,0) -- (0,8pt);}}}},
		]
		\draw[->] (0,0) -- (6.5,0) node[right] {$x_1$};
		\draw[->] (0,0) -- (0,6) node[below left] {$x_2$};
		\path[name path = c1] (0,4) -- (4,0);
		\path[name path = c2] (0,1) -- (6,1);
		\path[name path = c3] (0,1) -- (2.5,6);
		\path[name path = c4] (0,5) -- (6,5);
		\path[name path = c5]  (6,0) -- (6,6);
		\path [name intersections={of=c1 and c3, by={p13}}];
		\path [name intersections={of=c3 and c4, by={p34}}];
		\path [name intersections={of=c4 and c5, by={p45}}];
		\path [name intersections={of=c2 and c5, by={p25}}];
		\path [name intersections={of=c1 and c2, by={p12}}];
		\fill[color=gray!25] (p13) -- (p34) -- (p45) -- (p25) -- (p12) -- cycle;
		\draw[lb] (0,4) -- (4,0) node[left, pos=0, xshift=1pt] {$c_1$};
		\draw[lb] (0,1) -- (6,1) node[above right, yshift=-2pt] {$c_2$};
	
		\draw[ub] (0,1) -- (2.25,5.5) node[left, xshift=2pt, yshift=1pt] {$c_3$};
		\draw[ub] (0,5) -- (6,5) node[below right] {$c_4$};
		
		\draw[color=black75, thick, dashed] (3,0) -- (3,6);
		\end{tikzpicture}
	\end{subfigure}	
\end{figure}

    The encoding of the $c_1$-case is given by $(\bnd{2}{c_2}\leq \bnd{2}{c_1})\land (\bnd{2}{c_1} \leq \bnd{2}{c_3}) \land (\bnd{2}{c_1} \leq \bnd{2}{c_4})$, which evaluates to $(x_1 \leq 3) \land (-3x_1 \leq -3) \land (-x_1 \leq 1)$ satisfied by any $x_1 \in [1,3]$, on the left of the dashed line.
    The case for $c_2$ evaluates to $(-x_1 \leq -3) \land (-2x_1 \leq 0) \land (0 \leq 4)$ and is satisfiable on the right of the dashed line. The disjunction of the two formulas then defines exactly those values for $x_1$ which allow a solution of the initial system.
\end{example}

\noindent The construction for the case $i \in \ubs{a}{j}$ designating $i$ as smallest upper bound is analogous.
In matrix representation, these projections are defined by the following transformation:

\begin{definition}[Restricted Projection] \label{def:restricted-projection}
	Let $\matr{A} \in \rationals^{m \times n}$, $\vect{b} \in \rationals^{m\times 1}$ and $j \in [n]$.
	\begin{itemize}
		\item If $\lbs{a}{j}\not=\emptyset$ and $\ubs{a}{j}\not=\emptyset$, then for any $i \in \lbs{a}{j} \cup \ubs{a}{j}$ we fix $\matr{f} \in \rationals^{(m-1) \times m}$ arbitrarily but deterministically to consist of exactly the following rows:
		\begin{align*}
			\frac{1}{a_{i,j}} \cdot \rvect{e^{(m)}}{i} - \frac{1}{a_{i',j}} \cdot \rvect{e^{(m)}}{i'} &\text{ for every } i' \in \lbs{a}{j} \setminus \{i\},\\
			-\frac{1}{a_{i,j}} \cdot \rvect{e^{(m)}}{i} + \frac{1}{a_{i',j}} \cdot \rvect{e^{(m)}}{i'} &\text{ for every } i' \in \ubs{a}{j} \setminus \{i\},\qquad \text{ and }\qquad \rvect{e^{(m)}}{i'} \text{ for every } i' \in \nbs{a}{j}.
		\end{align*} 
		Then the \emph{restricted projection} $\fmpproj{\sys{a}{b}}{j}{i}$ of $x_j$ w.r.t. the row $i$ from the system $\sys{a}{b}$ is defined as the system $\matr{f}\matr{a}\vect{x}\leq \matr{f} \vect{b}$.
		We call $\matr{f}$ the \emph{projection matrix} corresponding to $\fmpproj{\sys{a}{b}}{j}{i}$.
		\item If $\lbs{a}{j}=\emptyset$ or $\ubs{a}{j}=\emptyset$, then we define the projection matrix $\matr{f} \in \rationals^{|\nbs{a}{j}| \times m}$ to have exactly one row $\rvect{e^{(m)}}{i'}$ for each $i' \in \nbs{a}{j}$, and define $\fmpproj{\sys{a}{b}}{j}{\bot}$ as $\matr{f}\matr{a}\vect{x}\leq \matr{f} \vect{b}$.
	\end{itemize}
\end{definition}

\noindent The following lemma states a crucial result for our method: The solutions of the restricted projections for all lower (or all upper) bounds of a variable exactly cover the projection of the entire solution set.

\begin{lemma} \label{lma:restricted-projection}
	Let $\matr{A} \in \rationals^{m \times n}$, $\vect{b} \in \rationals^{m\times 1}$, $j \in [n]$ and $I \in \{ \lbs{a}{j}, \ubs{a}{j}\}$.
	If $\lbs{a}{j}\neq \emptyset$ and $\ubs{a}{j}\neq \emptyset$, then
	\[
		\proj{\solset{\sys{a}{b}}}{X\setminus\{x_j\}} = \bigcup_{i \in I} \solset{\fmpproj{\sys{a}{b}}{j}{i}}.
	\]
	Otherwise ($\lbs{a}{j}=\emptyset$ or $\ubs{a}{j}=\emptyset$), it holds $\proj{\solset{\sys{a}{b}}}{X\setminus\{x_j\}} = \solset{\fmpproj{\sys{a}{b}}{j}{\bot}}$.
\end{lemma}
\begin{proof}
	The case $\lbs{a}{j}=\emptyset$ or $\ubs{a}{j}=\emptyset$ follows from the correctness of FM.
	Assume $I = \lbs{a}{j}$, the case $I = \ubs{a}{j}$ is analogous.
	\begin{description}
		\item[$\supseteq$:] Let $i \in \lbs{a}{j}$ and $\alpha \models \fmpproj{\sys{a}{b}}{j}{i}$, then for all $\ell \in \lbs{a}{j}$, $u \in \ubs{a}{j}$ it holds $\alpha(\bnd{j}{\rvect{a}{\ell}\;\vect{x} \leq b_\ell}) \leq \alpha(\bnd{j}{\rvect{a}{i}\;\vect{x} \leq b_i}) \leq \alpha(\bnd{j}{\rvect{a}{u}\;\vect{x} \leq b_u})$.
Thus, $\alpha[x_j \mapsto \alpha(\bnd{j}{\rvect{a}{i}\;\vect{x} \leq b_i})] \models \sys{a}{b}$.

		\item[$\subseteq$:] Let $\alpha \models \sys{a}{b}$ and $i = \arg\max_{\ell \in \lbs{a}{j}}(\alpha(\bnd{j}{\rvect{a}{\ell}\;\vect{x} \leq b_\ell}) )$, then for all $u \in \ubs{a}{j}$ it holds\\
		$\alpha(\bnd{j}{\rvect{a}{i}\;\vect{x} \leq b_i}) \leq \alpha(\bnd{j}{\rvect{a}{u}\;\vect{x} \leq b_u})$ and thus $\alpha \models \fmpproj{\sys{a}{b}}{j}{i}$.\qedhere
	\end{description}
\end{proof}

\begin{definition}[FMplex Variable Elimination]
	For $\matr{A} \in \rationals^{m \times n}$, $\vect{b} \in \rationals^{m\times 1}$, $j \in [n]$ and $* \in \{ -,+ \}$, we define
	\begin{align*}
		\fmplexelim^*_{j}(\sys{a}{b}) = & \left.\begin{cases}
		\{\fmpproj{\sys{a}{b}}{j}{i} \mid i \in \sbs{a}{j}\} & \text{ if } \lbs{a}{j} \neq \emptyset \textit{ and }  \ubs{a}{j}\neq\emptyset\\
		\{ \fmpproj{\sys{a}{b}}{j}{\bot}  \} &  \textit{otherwise}.
		\end{cases}\right.
	\end{align*}
\end{definition}

The FMplex elimination defines a set of restricted projections which can be composed to the full projection according to \Cref{lma:restricted-projection}. Lifting this from sets to logic naturally results in the following theorem which demonstrates the usage of our method.

\begin{theorem}\label{thm:correctness-elim}
	Let $\matr{A} \in \rationals^{m \times n}$, $\vect{b} \in \rationals^{m\times 1}$, and $j \in [n]$.
	Then
	\[\exists x_j.\; \sys{a}{b}\quad \equiv \quad\bigvee\nolimits_{S\in \fmplexelim^+_{j}(\sys{a}{b})} S  \quad\equiv\quad \bigvee\nolimits_{S\in \fmplexelim^-_{j}(\sys{a}{b})} S.\] 
\end{theorem}

For eliminating multiple variables, we iteratively apply $\fmplexelim^-$ or $\fmplexelim^+$ to each restricted projection resulting from the previous elimination step.
Note that we can choose the next variable to be eliminated as well as the variant independently in every branch.

\begin{example}\label{example:fmplex-elim-procedure}
	We continue \Cref{example:fmplex-idea}, from which we eliminated $x_2$ and now want to eliminate $x_1$: 
	\begin{align*}
		\exists x_1. \; \exists x_2.\; \sys{a}{b}\ &\equiv\ \exists x_1.\bigvee\nolimits_{S\in \fmplexelim^-_{2}(\sys{a}{b})} S\\
		&\equiv \;\exists x_1.\; \left(x_1 \leq 3 \land -3x_1 \leq -3\land -x_1 \leq 1\right)\; \lor  \exists x_1. \;\left(-x_1 \leq -3 \land -2 x_1 \leq 0 \land 0 \leq 4\right) 
	\end{align*}
	We eliminate the two quantifiers for $x_1$ separately, using
	\begin{align*}
		\fmplexelim^-_{1}\left(x_1 \leq 3 \land -3x_1 \leq -3\land -x_1 \leq 1\right) &= \{(0 \leq 2 \land 0 \leq 2), (0 \leq -2\land 0 \leq 4 )\}\ \text{and}\\
		\fmplexelim^-_{1}\left(-x_1 \leq -3 \land -2 x_1 \leq 0 \land 0 \leq 4\right) &= \{(0 \leq 4)\}
	\end{align*}
	giving us the final result $\exists x_1. \; \exists x_2.\; \sys{a}{b}\ \equiv\ ((0 \leq 2 \land 0 \leq 2) \lor (0 \leq 4 \land 0 \leq -2)) \lor (0 \leq 4)$.
\end{example}
	
We analyze the complexity in terms of the number of new rows (or constraints) that are constructed during the elimination of all variables:

\begin{theorem}[Complexity of \fmplexelim]
	Let $\matr{A} \in \rationals^{m \times n}$, and $\vect{b} \in \rationals^{m\times 1}$.
	When eliminating $n$ variables from $\sys{a}{b}$, the $\fmplexelim^-$ method constructs $\mathcal{O}(n\cdot m^{n+1})$ new rows.
\end{theorem}
\begin{proof}
	The number $N(m,n)$ of constructed rows is maximal if the system consists only of lower bounds and one upper bound.
	Then, $\fmplexelim^-$ yields $m-1$ new systems of size $m-1$, from which $n-1$ variables need to be eliminated; thus $N(m,n) \leq (m-1)\cdot((m-1) + N(m-1,n-1))$.
	With $k = \min(n,m)$, we obtain
	$N(m,n) \leq \sum\limits_{i=1}^{k}(m-i)\cdot \prod\limits_{j=1}^{i}(m-j) \leq n\cdot m^{n+1}$.
\end{proof}

While still exponential, this bound is considerably better than the theoretical doubly exponential worst-case complexity of the FM method.
Shortly speaking, FMplex trades one exponential step at the cost of the result being a decomposition into multiple partial projections.
However, there are systems for which FMplex produces strictly more rows than the FM method: In the worst case from the above proof, FM obtains a single system of the same size as each of the sub-problems computed by $\fmplexelim^-$.
Although in this case, we could simply employ $\fmplexelim^+$ instead, it is unclear whether there exists a rule for employing $\fmplexelim^-$ or $\fmplexelim^+$ that never produces more constraints than FM.

Like FM, FMplex keeps redundancies from the input throughout the algorithm, thus there might be identical rows in the same or across different sub-problems.
But in contrast to FM, FMplex does not introduce any redundancies by construction in the sense of \Cref{def:redundancyconstruction}.

\begin{theorem}
	Let $\matr{A} \in \rationals^{m \times n}$, $\vect{b} \in \rationals^{m\times 1}$ and $k\in[m]$.
	Assume $(\sys{a^{(0)}}{b^{(0)}}) = (\sys{a}{b})$ and for all $j \in [k]$, let $(\sys{a^{(j)}}{b^{(j)}}) \in \fmplexelim^-_{j}(\sys{a^{(j-1)}}{b^{(j-1)}}) \cup \fmplexelim^+_{j}(\sys{a^{(j-1)}}{b^{(j-1)}})$.
Let $\matr{F^{(1)}}, \ldots, \matr{F^{(k)}}$ be the respective projection matrices, and $\matr{F} = \matr{F^{(k)}} \cdot \ldots \cdot \matr{F^{(1)}}$.
	Then $F$ is linearly independent.
\end{theorem}
\begin{proof}
  By definition, the projection matrices are linearly independent, and thus so is their product $\matr{F}$.
\end{proof}

\section{FMplex as Satisfiability Checking Procedure}
\label{sec:sat}

A formula is satisfiable if and only if eliminating all variables (using any quantifier elimination method such as FM or FMplex) yields a tautology.
However, FMplex computes smaller sub-problems whose satisfiability implies the satisfiability of the original problem.
Therefore, we do not compute the whole projection at once, but explore the decomposition using a depth-first search.
The resulting search tree has the original system as root, and each node has as children the systems resulting from restricted projections.
The original system is satisfiable if and only if a leaf without any trivially false constraints exists.
An example is depicted in \Cref{figure:example-searchtree}.
We start with a basic version of the algorithm and then examine how the search tree can be pruned, resulting in two variants; all versions are given in \Cref{algo:combined}.

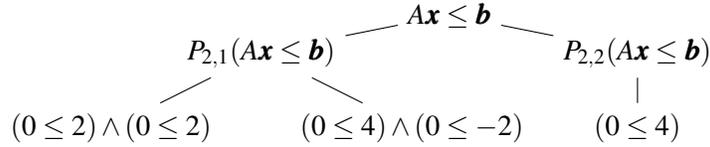
\begin{figure}[h!]
	\centering
	\begin{tikzpicture}
		\tikzstyle{level 1} = [level distance=5mm, sibling distance = 50mm]
		\tikzstyle{level 2} = [level distance=10mm, sibling distance = 40mm]
			\node[] (root) {$\sys{a}{b}$} 
				child{node[] (v11) {$\fmpproj{\sys{a}{b}}{2}{1}$}
					child{node[] (v21) {$(0 \leq 2)\land (0 \leq 2)$}}
					child{node[] (v22) {$(0 \leq 4) \land (0 \leq -2)$}}
				}
				child{node[] (v12) {$\fmpproj{\sys{a}{b}}{2}{2}$}
					child{node[] (v23) {$(0 \leq 4)$}}
				};
		\end{tikzpicture}
	
	\caption{%
	The search tree corresponding to \Cref{example:fmplex-elim-procedure}.
	The very first leaf (bottom left) is already satisfiable, meaning that the rest would not need to be computed.%
	}\label{figure:example-searchtree}
\end{figure}

An important observation is that we can decide independently for each node of the search tree, which variable to eliminate next and whether to branch on lower or on upper bounds.

\begin{definition}[Branch Choices]
	The set of \emph{branch choices} for a system $\sys{a}{b}$ is
	\begin{align*}
		\textit{branch\_choices}(\sys{a}{b}) = & \{ \{ (x_j, i) \mid i \in \lbs{a}{j} \} \mid j \in [n]\wedge \lbs{a}{j} \neq \emptyset \wedge  \ubs{a}{j} \neq \emptyset \} \\
		\cup & \{ \{ (x_j, i) \mid i \in \ubs{a}{j} \} \mid j \in [n]\wedge \lbs{a}{j} \neq \emptyset \wedge \ubs{a}{j} \neq \emptyset \} \\
		\cup & \{ \{ (x_j, \bot) \} \mid j \in [n]\wedge (\lbs{a}{j} = \emptyset \vee \ubs{a}{j} = \emptyset) \}.
	\end{align*}
\end{definition}

For an initial input $\sys{\widehat{a}}{\widehat{b}}$ with $\widehat{m}$ rows, we define the depth-first search using the recursive method $\texttt{FMplex}(\sys{\widehat{a}}{\widehat{b}};\sys{a}{b},\matr{f})$ in \refalgobase where $\sys{a}{b}$ is the currently processed sub-problem in the recursion tree.
We track the relation of $\sys{a}{b}$ to $\sys{\widehat{a}}{\widehat{b}}$ in terms of linear combinations using the parameter $\matr{f}$.
The initial call is defined as $\texttt{FMplex}(\sys{\widehat{a}}{\widehat{b}}) = \texttt{FMplex}(\sys{\widehat{a}}{\widehat{b}};\sys{\widehat{a}}{\widehat{b}},\matr{e}^{(\widehat{m})})$.
We allow that $\sys{a}{b}$ contains identical rows when they are obtained in different ways (which is reflected by $\matr{f}$).
We need to keep these duplicates for proving the results of this section.

\vspace{-0.5em}\paragraph{Solutions} If a trivially satisfiable node is found, the algorithm constructs an assignment starting with the empty assignment and extends it in reverse order in which the variables were eliminated.
For every variable $x_j$, a value is picked above all lower and below all upper bounds on $x_j$ evaluated at the underlying assignment.
By the semantics of the projection, the value of the designated (largest lower or smallest upper) bound on $x_j$ is suitable.

\vspace{-0.5em}\paragraph{Conflicts} We distinguish inconsistencies in $\sys{a}{b}$ by the following notions: 
We call a row $i$ of $\sys{a}{b}$ a \emph{conflict} if it is of the form $\rvect{a}{i} = \vect{0}^{(n)}$ with $b_i < 0$.
We call the conflict \emph{global} if $\rvect{f}{i} \geq 0$  and \emph{local} otherwise.
In case of a global conflict, Farkas' Lemma allows to deduce the unsatisfiability of $\sys{\widehat{a}}{\widehat{b}}$, thus stopping the search before the whole search tree is generated.
Then a set of conflicting rows $K$ of the input system corresponding to $\rvect{f}{i}$ is returned.
In particular, the set $\{\rvect{\widehat{a}}{j}\; \vect{x} \leq \widehat{b}_j \mid f_{i,j} \neq 0\}$ is a minimal unsatisfiable subset of the constraints in $\sys{\widehat{a}}{\widehat{b}}$.
In case of a local conflict, we simply continue to explore the search tree.
The algorithm returns $\partialunsat$ to indicate that $\sys{a}{b}$ is unsatisfiable, but the unsatisfiability of $\sys{\widehat{a}}{\widehat{b}}$ cannot be derived.
This approach, formalized in \refalgobase\!\!, guarantees that the initial call will never return $\partialunsat$; we always find either a global conflict or a solution.

\vspace*{0.5em}\noindent The correctness and completeness of \texttt{FMplex} follows from \Cref{thm:correctness-elim} and \Cref{thm:global-conflict}.

\begin{theorem} \label{thm:global-conflict}
	Let $\matr{\widehat{a}} \in \rationals^{\widehat{m} \times n}$, and $\vect{\widehat{b}} \in \rationals^{\widehat{m}}\times 1$.
	Then $\sys{\widehat{a}}{\widehat{b}}$ is unsatisfiable if and only if the call $\texttt{FMplex}(\sys{\widehat{a}}{\widehat{b}})$ to \refalgobase terminates with a global conflict.
\end{theorem}

  \begin{algorithm}[p]%
    \caption{$\texttt{FMplex}(\sys{\widehat{a}}{\widehat{b}};\sys{a}{b},\matr{f},$\chng{$\mathcal{N},I$}$,$\chngd{$\texttt{lvl},\texttt{bt\_lvl}$}$)$}\label{algo:combined}
      \begin{description}
		\item[Algorithm 1a] \label{algo:fmplex1} The base method consists of the plain (unframed and unfilled) parts.
		\item[Algorithm 1b] \label{algo:leaving-out-bounds} Consists of the base method and the \chng{framed parts}. 
		\item[Algorithm 1c] \label{algo:backtracking} Consists of the base method, the \chng{framed parts} and the \chngd{filled boxes}.
	\end{description}
	\hrulefill\\
	\SetKwInOut{Data}{Data}
    \SetKwInOut{Input}{Input}
    \SetKwInOut{Output}{Output}
    \DontPrintSemicolon
    \SetKwFunction{FCheck}{FMplex}
    \SetKwProg{Fn}{Function}{:}{}
	\Data{$\matr{\widehat{a}} \in \rationals^{\widehat{m} \times n},\ \vect{\widehat{b}}\in \rationals^{\widehat{m}}$}
    \Input{$\matr{a} \in \rationals^{m \times n},\ \vect{b}\in \rationals^{m}$, $\matr{f} \in \rationals^{m \times \widehat{m}}$ s.t. $\matr{f}\matr{\widehat{a}} = \matr{a}$ and $\matr{f}\vect{\widehat{b}} = \vect{b}$, \chng{$\mathcal{N}\subseteq [\widehat{m}]$, $I\subseteq [\widehat{m}]$}, \chngd{$\texttt{lvl} \in [n] \cup \{ 0 \}$, and $\texttt{bt\_lvl}: [m] \to [n] \cup \{ 0 \}$}}
    \Output{(\sat, $\alpha$) with $\alpha \models \matr{a}\vect{x} \leq \vect{b}$, or (\unsat, $K$) where $K \subseteq [\widehat{m}]$, or\\
	(\partialunsat, \chngd{$l, K$}) \chngd{where $l \in [n]$ and $K \subseteq [\widehat{m}]$}}
	\BlankLine
	\lIf{$\matr{a} = 0 \land \vect{b} \geq 0$} {
		\Return{$($\sat, $())$} 
	}
    \lIf{$\exists i \in [m].\; \rvect{a}{i} = 0 \land b_i < 0 \land \rvect{f}{i} \geq 0$}{
            \Return{$(\unsat, \{ i' \mid f_{i,i'} \neq 0 \})$} 
    } 
	\boxchngd{a}
	\boxchngdstart{a}
	\If{$\exists i \in [m].\; \rvect{a}{i} = 0 \land b_i < 0 \land \rvect{f}{i} \ngeq 0$}{
		$i := \arg\min_{i \in [m]}\{ \texttt{bt\_lvl}(i) \mid \rvect{a}{i} = 0 \land b_i < 0 \}$\;
		\Return $(\partialunsat,\texttt{bt\_lvl}(i)-1,\{ i' \mid f_{i,i'} \neq 0 \})$ \;
	}
	\boxchngdend{a}%
	\chngd{$K = \emptyset$} \;
	\textbf{choose} $V \in \textit{branch\_choices}(\matr{a}\vect{x} \leq \vect{b},\chng{$\{ \mathcal{B}^{-1}_{\mathcal{N},F}(i) \mid i \in I \}$})$\;
    \ForEach{$(x_j, i) \in V$}{
		\textbf{compute} $\matr{a'}\vect{x} \leq \vect{b'} := \fmpproj{\matr{a}\vect{x} \leq \vect{b}}{j}{i}$ with projection matrix $\matr{f'}$ \chngd{and backtrack levels $\texttt{bt\_lvl}'$} \;
		\chng{$\mathcal{N}'$ := $\mathcal{N} \cup \{ \mathcal{B}_{\mathcal{N},\matr{f}}(i) \}$ \textbf{ if }  $i \neq \bot$ \textbf{ else }   $\mathcal{N}$}\;
        \Switch{\FCheck($\sys{\widehat{a}}{\widehat{b}};\matr{a'}\vect{x} \leq \vect{b'},\matr{f'}\matr{f},\chng{$\mathcal{N}',I$},\chngd{$\texttt{\upshape lvl}+1,\texttt{\upshape bt\_lvl}'$}$)}{
			\lCase{$(\unsat, K')$}{\Return{$(\unsat, K')$}}
        	\lCase{$(\sat, \alpha)$}{\Return{$(\sat, \alpha[x_j \mapsto r])$} for a suitable $r \in \rationals$}
			\boxchngd{b}
			\boxchngdstart{b}
			\Case{$(\partialunsat, l, K')$}{
				\lIf{$l<\texttt{\upshape lvl}$}{
					\Return{$(\partialunsat, l, K')$} %
				}
				\lElse {
					$K = K \cup K'$ %
				}
			}
    	}
		\boxchngdend{b}	\chng{$I$ := $I \cup \{ \mathcal{B}_{\mathcal{N},\matr{f}}(i) \}$} \;\label{algo:backtracking:ignore}
    }
	\chngd{
	\lIf{$\texttt{\upshape lvl}=0$}{
		\Return{$(\unsat, K)$}
	}}\;
	\Return{$(\partialunsat, \chngd{\texttt{\upshape lvl-1}, K})$}
  \end{algorithm}

\begin{proof}[Proof Idea for \Cref{thm:global-conflict}]
	If $\sys{\widehat{a}}{\widehat{b}}$ is unsatisfiable, then there exists a minimal unsatisfiable subset $\widehat{K}$ of the corresponding constraints. We construct a path in the search tree induced by \refalgobase yielding a conflict that is a linear combination of $\widehat{K}$. As $\widehat{K}$ is minimal, the linear combination is positive, i.e. the conflict is global.
	The other direction of the equivalence follows immediately with Farkas' Lemma.
	Consult the
	\ifextended
	appendix
	\else
	extended version
	\fi
	for a detailed proof.
\end{proof}

\subsection{Avoiding Redundant Checks}\label{subsec:avoiding-redundancies}

We observe that each row $i$ in a sub-problem $\sys{a}{b}$ in the recursion tree of $\texttt{FMplex}(\sys{\widehat{a}}{\widehat{b}})$ corresponds to a row $\hat{\imath}$ in $\sys{\widehat{a}}{\widehat{b}}$ in the sense that it is a linear combination of the rows $\{ \hat{\imath} \} \cup \mathcal{N}$ of $\sys{\widehat{a}}{\widehat{b}}$, where $\mathcal{N} \subseteq [\widehat{m}]$ corresponds to the lower/upper bounds designated as largest/smallest one to compute $\sys{a}{b}$:

\begin{theorem} \label{thm:basis-nonbasis-correspondence}
	Let $\matr{\widehat{a}} \in \rationals^{\widehat{m} \times n}$ and $\vect{\widehat{b}} \in \rationals^{\widehat{m}\times 1}$.
	Let $\texttt{FMplex}(\sys{\widehat{a}}{\widehat{b}};\matr{{a}}\vect{x} \leq \vect{{b}},\matr{f})$ be a call in the recursion tree of the call $\texttt{FMplex}(\sys{\widehat{a}}{\widehat{b}})$ to \refalgobase, where $\matr{A}\in \rationals^{m\times n}$ and $\vect{b} \in \rationals^{m\times 1}$ (by construction $m\leq \widehat{m}$).

	Then there exists a set $\mathcal{N} \subseteq [\widehat{m}]$ such that
	\begin{enumerate}
		\item $\sys{a}{b}$ is satisfiable if and only if $(\sys{\widehat{a}}{\widehat{b}}) \wedge (\submatr{\widehat{a}}{\mathcal{N}}\vect{x} = \subvect{\widehat{b}}{\mathcal{N}})$ is satisfiable,
		\item there exists an injective mapping $\mathcal{B}_{\mathcal{N}, F}: [m] \to [\hat{m}], i \mapsto \hat{\imath}$ with
		$\{ \hat{\imath} \} = \{ i' \in [\hat{m}] \mid f_{i,i'} \neq 0 \} \setminus \mathcal{N}$.
	\end{enumerate}
\end{theorem}
\begin{proof}[Proof Idea]
	The statement follows with a straight forward induction over the elimination steps, where the original row corresponding to the chosen bound is added to $\mathcal{N}$, and $\mathcal{B}_{\mathcal{N}, F}$ keeps track of which constraint corresponds to which original row.
	Consult the
	\ifextended
	appendix
	\else
	extended version
	\fi
	for a detailed proof.
\end{proof}

We call the above defined set $\mathcal{N}$ the \emph{non-basis}, inspired from the analogies to the simplex algorithm (discussed in \Cref{sec:simplex}).
By the above theorem, the order in which a non-basis is constructed has no influence on the satisfiability of the induced sub-problem.
In particular:

\begin{theorem} \label{thm:redundancies}
	Let $\matr{a} \in \rationals^{m \times n}$, $\vect{b} \in \rationals^{m\times 1}$, $j \in [n]$, and let $i,i' \in [m]$ be row indices with $a_{i,j} \neq 0$ and $a_{i',j} \neq 0$.
	If $\fmpproj{\matr{a}\vect{x} \leq \vect{b}}{j}{i}$ is unsatisfiable, then $\fmpproj{\matr{a}\vect{x} \leq \vect{b}}{j}{i'} \wedge  (\rvect{a}{i}\;\vect{x} = b_i) $ is unsatisfiable.
\end{theorem}
\begin{proof}
	By \Cref{thm:basis-nonbasis-correspondence}, if $\fmpproj{\matr{a}\vect{x} \leq \vect{b}}{j}{i}$ is unsatisfiable, then $(\matr{a}\vect{x} \leq \vect{b}) \wedge  (\rvect{a}{i}\;\vect{x} = \vect{b_i})$ is unsatisfiable, and trivially $(\matr{a}\;\vect{x} \leq \vect{b}) \wedge  (\rvect{a}{i}\;\vect{x} = \vect{b_i}) \wedge (\rvect{a}{i'}\;\vect{x} = \vect{b_{i'}})$ is unsatisfiable as well.
	Using \Cref{thm:basis-nonbasis-correspondence} in the other direction yields that $\fmpproj{\matr{a}\vect{x} \leq \vect{b}}{j}{i'} \wedge (\rvect{a}{i}\;\vect{x} = \vect{b_i} )$ is unsatisfiable.
\end{proof}

This suggests that if $\texttt{FMplex}(\sys{\widehat{a}}{\widehat{b}};\sys{a}{b},\matr{f})$ with non-basis $\mathcal{N}$ has a child call for row $i$ which does not return $\sat$, then no other call in the recursion tree of $\texttt{FMplex}(\sys{\widehat{a}}{\widehat{b}};\sys{a}{b},\matr{f})$ where the corresponding non-basis contains $\mathcal{B}_{\mathcal{N},F}(i)$ will return $\sat$ either.
Hence, we can ignore $\mathcal{B}_{\mathcal{N},F}(i)$ as designated bound in the remaining recursion tree of $\texttt{FMplex}(\sys{\widehat{a}}{\widehat{b}};\sys{a}{b},\matr{f})$.

\begin{example}
	Consider the system from \Cref{example:fmplex-idea}, with an additional constraint $c_5 : (-x_2 \leq 0)$.
	If $c_5$ is tried first as greatest lower bound on $x_2$, then the combination with $c_2 : (-2x_2 \leq -2)$ yields the local conflict $\frac{1}{2}c_2 - c_5 = (0 \leq -1)$.
	Thus, this branch and, due to \Cref{thm:redundancies}, any non-base containing row $5$ yields an unsatisfiable system.

	Next, we try $c_1$ as greatest lower bound on $x_2$ resulting in the combinations $\frac{1}{2}c_2 - c_1 = (x_1 \leq 3)$, $c_5 - c_1 = (x_1 \leq 4)$, $c_1 + c_3 = (-3x_1 \leq -3)$ and $c_1 + c_4 = (-x_1 \leq 1)$ and corresponding non-base $\{1\}$.

	If we now choose $(x_1 \leq 4)$ as smallest upper bound on $x_1$, leading to the non-base $\{1,5\}$, another local conflict occurs: $(x_1 \leq 3) - (x_1 \leq 4) = (0 \leq -1)$.
	As $5$ is contained in the non-base, we could know beforehand that this would happen and thus avoid computing this branch.
\end{example}

We update the $\texttt{FMplex}$ algorithm as shown in \refalgobounds using the following definition:

\begin{definition}
	The set of \emph{branch choices} for $\sys{a}{b}$ with $m$ rows w.r.t. $I \subseteq [m]$ is
	\begin{align*}
		\textit{branch\_choices}(\sys{a}{b}, I) =\quad & \ \{ \{ (x_j, i) \mid i \in \lbs{a}{j} \setminus I  \} \mid j \in [n]\wedge \lbs{a}{j} \neq \emptyset \wedge \ubs{a}{j} \neq \emptyset\} \\
		\cup  &\ \{ \{ (x_j, i) \mid i \in \ubs{a}{j} \setminus I \} \mid j \in [n]\wedge \lbs{a}{j} \neq \emptyset\wedge \ubs{a}{j} \neq \emptyset \} \\
		\cup &\ \{ \{ (x_j, \bot) \} \mid j \in [n]\wedge (\lbs{a}{j} = \emptyset \vee \ubs{a}{j} = \emptyset) \}.
	\end{align*}
\end{definition}

\noindent It is easy to see that this modification prevents visiting non-basis twice in the following sense:

\begin{theorem}
	Let $\texttt{FMplex}(\sys{\widehat{a}}{\widehat{b}};\matr{a}\vect{x} \leq \vect{b},\_,\mathcal{N},\_)$ and $\texttt{FMplex}(\sys{\widehat{a}}{\widehat{b}};\matr{a'}\vect{x} \leq \vect{b'},\_,\mathcal{N}',\_)$ be two calls in the recursion tree of a call to \refalgobounds\!\!.
	Then either $\mathcal{N} \neq \mathcal{N'}$ or one of the systems occurs in the subtree below the other and only unbounded variables are eliminated between them (i.e. one results from the other by deleting some rows).\qed
\end{theorem}

\Cref{thm:unique-base-termination} states that, still, \refalgobounds always terminates with $\sat$ or a global conflict. This follows by a slight modification of the proof of \Cref{thm:global-conflict}, presented in the
\ifextended
appendix
\else
extended version
\fi
of this paper.

\begin{theorem}\label{thm:unique-base-termination}
	Let $\matr{\widehat{a}} \in \rationals^{\widehat{m} \times n}$, and $\vect{\widehat{b}} \in \rationals^{\widehat{m}\times 1}$.
	Then $\sys{\widehat{a}}{\widehat{b}}$ is unsatisfiable if and only if the call $\texttt{FMplex}(\sys{\widehat{a}}{\widehat{b}})$ to \refalgobounds terminates with a global conflict.\qed
\end{theorem}

\subsection{Backtracking of Local Conflicts}

So far, we ignored local conflicts that witness the unsatisfiability of a given sub-problem.
In this section, we will cut off parts of the search tree based on local conflicts and examine the theoretical implications.

We applied Farkas' Lemma on conflicting rows in some sub-problem that are positive linear combinations of rows from the input system.
We can also apply Farkas' Lemma to conflicting rows which are positive linear combinations of some \emph{intermediate} system to conclude the unsatisfiability of the latter.
Whenever such a conflict occurs, we can backtrack to the parent system of that unsatisfiable system.
Instead of tracking the linear combinations of every row in terms of the rows of each preceding intermediate system, we can do an incomplete check: If a conflicting row was computed only by addition operations, then it is a positive linear combination of the involved rows.
Thus, we assign to every intermediate system a level, representing its depth in the search tree and store for every row the level where the last subtraction was applied to the row (i.e. a lower (upper) bound was subtracted from another lower (upper) bound).
If a row is conflicting, we can conclude that the intermediate system at this level is unsatisfiable, thus we can jump back to its parent.

Assume the current system is $\matr{a}\vect{x} \leq \vect{b}$ at level $\texttt{lvl}$ with $m$ rows whose backtracking levels are stored in $\texttt{bt\_lvl}: [m] \to ([n] \cup \{ 0 \})$.
If $\texttt{lvl}=0$, then $\texttt{bt\_lvl}$ maps all values to $0$.
When computing $\fmpproj{\matr{a}\vect{x} \leq \vect{b}}{j}{i}$ for some $x_j$ and $i$ with projection matrix $\matr{f}$, the backtracking levels of the rows in the resulting system $\matr{FA}\vect{x} \leq \matr{F}\vect{b}$ are stored in $\texttt{bt\_lvl}'$ where for each row $i''$
\[\texttt{bt\_lvl}'(i'') := \left.\begin{cases}
	\max\{\texttt{bt\_lvl}(i),\texttt{bt\_lvl}(i')\} & \text{ if } f_{i'',i}, f_{i'',i'} > 0 \text { and } f_{i'',k}=0,\  k \notin \{ i,i' \} \\
	\texttt{lvl} &  \text{ otherwise} .
\end{cases}\right.\]

The backtracking scheme is given in \refalgobacktracking\!, which returns additional information in the $\partialunsat$ case, that is the backtrack level $l$ of the given conflict, and a (possibly non-minimal) unsatisfiable subset $K$.

\begin{theorem}
	Let $\texttt{FMplex}(\_;\matr{a}\vect{x} \leq \vect{b},\_,\_,\_,\texttt{\upshape lvl},\_)$ be a call to \refalgobacktracking\!, and consider a second call $\texttt{FMplex}(\_;\matr{a'}\vect{x} \leq \vect{b'},\_,\_,\_,\_,\texttt{\upshape bt\_lvl}')$ in the recursion tree of the first call.
	If $\matr{a'}\vect{x} \leq \vect{b'}$ has a local conflict in a row $i$ with $\texttt{\upshape bt\_lvl}'(i)=\texttt{\upshape lvl}$, then $\matr{a}\vect{x} \leq \vect{b}$ is unsatisfiable.
\end{theorem}
\begin{proof}
	By construction of $\texttt{bt\_lvl}$', $\rvect{a'}{i}\;\vect{x} \leq b'_i$ is a positive sum of rows from $\matr{a}\vect{x} \leq \vect{b}$, i.e. there exists an $\vect{f} \in \rationals^{1\times m}$ such that $(\vect{f} \matr{a} \vect{x} \leq  \vect{f} \vect{b}) = (\rvect{a'}{i}\;\vect{x} \leq b'_i)$.
	Then by Farkas' Lemma, $\matr{a}\vect{x} \leq \vect{b}$ is unsatisfiable.
\end{proof}

While it is complete and correct, \refalgobacktracking does not always terminate with a \emph{global} conflict (i.e. \Cref{thm:global-conflict} does not hold any more), even if we do not ignore any rows (i.e. omit Line~\ref{algo:backtracking:ignore}):

\vspace*{1em}

\begin{example}
	\opencutright
	\renewcommand*{\windowpagestuff}{
		\vspace*{-1em}
		\begin{align*}
			\left[ {\begin{array}{ccc}
				0 & 0 & -1 \\
				1 & -1 & -1 \\
				1 & 0 & 0 \\
				-1 & 1 & 0 \\
				0 & -1 & 1
			\end{array} } \right] 
			\cdot
			\left[ {\begin{array}{ccc}
				x_1 \\ x_2 \\ x_3
			\end{array} } \right] 
			\leq
			\left[ {\begin{array}{ccc}
				0 \\ 0 \\ -1 \\ -1 \\ 0
			\end{array} } \right] 
		\end{align*}
	}
	\begin{cutout}{0}{0.5\linewidth}{0pt}{7}
		We use \refalgobacktracking to eliminate variables with the static order $x_3, x_2, x_1$ from the system on the right, always branching on lower bounds.
		We first choose row $1$ as greatest lower bound on $x_3$.
		Rows $3$ and $4$ are retained as they do not contain $x_3$ and the combination of row $1$ with row $5$ is positive, so these constraints have backtrack level $0$.
	\end{cutout}
		\noindent The combination with row $2$ has backtrack level $1$ because both rows are lower bounds.
		Using this constraint as greatest lower bound on $x_2$ and combining it with row $4$ leads to a local conflict with backtrack level $1$.
		This means that the call at level $1$ is unsatisfiable and thus we backjump to level $0$.

		The second branch is visited, leading to the non-basis $\mathcal{N}=\{ 2, 5, 1 \}$ after three steps, where a \emph{local} conflict lets us backjump to level $0$ again.
		As there are no more lower bounds on $x_3$, the algorithm returns \unsat ~without finding a global conflict.
\end{example}

\section{Relation to Other Methods}
\label{sec:othermethods}

\subsection{Simplex Algorithm}
\label{sec:simplex}

The simplex method \cite{dantzig1998linear,lemke1954dual} is an algorithm for linear optimization over the reals and is able to solve \emph{linear programs}.
The \emph{general simplex} \cite{dutertre2006integrating} is an adaption for checking the satisfiability of systems of linear constraints.
We illustrate its idea for the weak case.

Remind that given a system $\sys{a}{b}$ with $m$ rows, by the fundamental theorem of linear programming (\Cref{thm:linearprogramming}), $\sys{a}{b}$ is satisfiable if and only if there exists some maximal subset $\mathcal{N} \subseteq [m]$ such that $\submatr{a}{\mathcal{N}}$ is linearly independent and $\sys{a}{b} \cup \submatr{a}{\mathcal{N}}\vect{x}=\subvect{b}{\mathcal{N}}$ is satisfiable - the latter can be checked algorithmically using Gaussian elimination, resulting in a system where each variable is replaced by bounds induced by the rows $\mathcal{N}$.
This system along with the information which element in $\mathcal{N}$ was used to eliminate which variable is called a \emph{tableau}.
The idea of the simplex method is to do a local search on the set $\mathcal{N}$ (called \emph{non-basis}), that is, we replace some $i \in \mathcal{N}$ (\emph{leaving variable}) by some $i' \in [m] \setminus \mathcal{N}$ (\emph{entering variable}) obtaining $\mathcal{N}' := \mathcal{N} \cup \{ i' \} \setminus \{ i \}$ such that $\submatr{a}{\mathcal{N}'}$ is still linearly independent.
The clou is that the tableau representing $(\sys{a}{b}) \wedge (\subsyseq{a}{b}{\mathcal{N}})$ can be efficiently transformed into $(\sys{a}{b}) \wedge (\subsyseq{a}{b}{\mathcal{N}'})$ (called \emph{pivot operation}), and progress of the local search can be achieved by the choice of $i$ and $i'$.
These local search steps are performed until a satisfying solution is found, or a conflict is found.
These conflicts are detected using Farkas' Lemma (\Cref{thm:farkas}), i.e. a row in the tableau induces a trivially false constraint and is a positive linear combination of some input rows.

As suggested by \Cref{thm:basis-nonbasis-correspondence}, there is a strong correspondence between a tableau of the simplex algorithm and the intermediate systems constructed in FMplex.
More precisely, if a non-basis of a simplex tableau is equal to the non-basis of a leaf system of \refalgobase\!, then the tableau is satisfiable if and only if the FMplex system is satisfiable.
In fact, we could use the same data structure to represent the algorithmic states.
Comparing the two algorithms structurally, FMplex explores the search space in a tree-like structure using backtracking, while simplex can jump between neighbouring leaves directly.

The idea for \refalgobounds that excludes visiting the same non-basis in fact arose from the analogies between the two methods.
Further, we observe a potential advantage of FMplex: Simplex has more non-bases reachable from a given initial state than the leaves of the search tree of FMplex, as FMplex needs only to explore all lower or all upper bounds of a variable while simplex does local improvements blindly.
Heuristically, simplex cuts off large parts of its search space and we expect it often visits fewer non-bases than FMplex - however, as the pruning done by FMplex is by construction of the algorithm, we believe that there might be combinatorially hard instances on which it is more efficient than simplex.

\subsection{Virtual Substitution Method}

\emph{Virtual substitution} \cite{loosApplyingLinearQuantifier1993,weispfenning1997quantifier} admits  quantifier elimination for real arithmetic formulas.
Here, we consider its application on existentially quantified conjunctions of linear constraints.

The underlying observation is that the satisfaction of a formula changes at the zeros of its constraints and is invariant between the zeros.
Thus, the idea is to collect all \emph{symbolic zeros} $\text{zeros}(\varphi)$ of all constraints in some input formula $\varphi$.
If all these constraints are weak, then a variable $x_j$ is eliminated by plugging every zero and an arbitrarily small value $-\infty$ into the formula, i.e. $\exists x_j.\; \varphi$ is equivalent to $\varphi[-\infty\vs x_j] \vee \bigvee_{\xi \in \text{zeros}(\varphi)} \varphi[\xi\vs x_j]$.
The formula $\varphi[t\vs x_j]$ encodes the semantics of substituting the term $t$ for $x_j$ into the formula $\varphi$ (which is a disjunction of conjunctions).
As we can pull existential quantifiers into disjunctions, we can iteratively eliminate multiple variables by handling each case separately.

The resulting algorithm for quantifier elimination is singly exponential; further optimizations (\cite{nipkowLinearQuantifierElimination2008a} even proposes to consider only lower or upper bounds for the test candidates) lead to a very similar procedure as the FMplex quantifier elimination: Substituting a test candidate into the formula is equivalent to computing the restricted projection w.r.t. a variable bound.
However, our presentation allows to exploit the correspondence with the FM method.

Virtual substitution can also be adapted for SMT solving \cite{vssmt} to a depth-first search similar to FMplex.
A conflict-driven search for virtual substitution on conjunctions of weak linear constraints has been introduced in \cite{korovin2014towards}, which tracks intermediate constraints as linear combinations of the input constraints similarly to FMplex. Their conflict analysis is a direct generalization of the global conflicts in FMplex and is thus slightly stronger than our notion of local conflicts. However, their method requires storing and maintaining a lemma database, while FMplex stores all the information for pruning the search tree locally. The approaches have strong similarities, although they originate from quite different methods. Further, our presentation shows the similarities to simplex, is easily adaptable for strict constraints, and naturally extensible to work incrementally.

\subsection{Sample-Based Methods}
There exist several depth-first search approaches, including McMillan et al. \cite{mcmillanGeneralizing2009}, Cotton \cite{cottonNatural2010} and Korovin et al. \cite{korovinConflict2009,korovinSolving2011}, which maintain and adapt a concrete (partial) variable assignment.
They share the advantage that combinations of constraints are only computed to guide the assignment away from an encountered conflict, thus avoiding many unnecessary combinations which FM would compute.

Similar to FMplex, these methods perform a search with branching, backtracking and learning from conflicting choices.
However, they branch on variable assignments, with infinitely many possible choices in each step.
Interestingly, the bounds learned from encountered conflicts implicitly partition the search space into a finite number of regions to be tried, similar to what FMplex does explicitly.
In fact, we deem it possible that \cite{korovinConflict2009} or \cite{korovinSolving2011} try and exclude assignments from exactly the same regions that FMplex would visit (even in the same order).
However, the sample-based perspective offers different possibilities for heuristic improvements than FMplex: choosing the next assigned value vs. choosing the next lower bound; deriving constant variable bounds vs. structural exploits using Farkas' Lemma; possibility of very quick solutions vs. more control and knowledge about the possible choices.

Moreover, these methods offer no straight-forward adaption for quantifier elimination, while FMplex does. However, \cite{mcmillanGeneralizing2009} and \cite{cottonNatural2010} can handle not only conjunctions, but any quantifier-free LRA formula in conjunctive normal form.

\section{Experimental Evaluation}
\label{sec:experiments}
We implemented several heuristic variants of the FMplex algorithm, as well as the generalized \emph{simplex} and the \emph{FM} methods as non-incremental DPLL(T) theory backends in our SMT-RAT solver \cite{corziliusSMTRAT2015} and compared their performance in the context of satisfiability checking.
Using the transformation given in \cite{nalbach2021extending} and case splitting as in \cite{barrett2006splitting}, we extended the method to also handle strict and not-equal-constraints.

The base version of FMplex (\refalgobase\!) was tested with two different heuristics for the choice of the eliminated variable and for the order in which the branches are checked.
These choices may strongly influence the size of the explored search tree; in the best case, the very first path leads to a satisfiable leaf or to a global conflict.

\vspace{-0.75em}\paragraph{Min-Fanout} We greedily minimize the number of children: for any $\sys{a}{b}$ and $I$, we choose $V \in \textit{branch\_choices}(\sys{a}{b},I)$ such that $\abs{V}$ is minimal; in case that this minimum is $1$, we prefer choices $V = \{(x_j,\bot)\}$ for a $j \in [n]$ over the other choices.
    
We prefer rows with a lower (earlier) backtrack level, motivated by finding a global conflict through trying positive linear combinations first.
Moreover, if backtracking is used then we expect this heuristic to allow for backtracking further back on average.

\vspace{-0.75em}\paragraph{Min-Column-Length} A state-of-the-art heuristic for simplex in the context of SMT solving is the \emph{minimum column length} \cite{soisimplex}: we choose the variables for leaving and entering the non-basis such that the number of necessary row operations is minimized.
We resemble this heuristic in FMplex as follows: we prefer choices $\{(x_j,\bot)\}$ and if there is no such $j$, we take the $j \in [n]$ with minimal $|\lbs{a}{j}|+|\ubs{a}{j}|$ and take the smallest choice between $\lbs{a}{j}$ and $\ubs{a}{j}$.

We first choose the rows which have the least non-zero coefficients (i.e. contain the least variables) to prefer sparse sub-problems.
This can be understood as \emph{Min-Row-Length}.

\vspace*{0.55em}
\noindent We consider the following solver variants: \texttt{FMplex-a-MFO} and \texttt{FMplex-a-MCL} implement \refalgobase with the Min-Fanout and the Min-Column-Length heuristic, respectively.
\texttt{FMplex-a-Rand-1/2} denotes two variants of \refalgobase where all choices are taken pseudo-randomly with different seeds.
\texttt{FMplex-b-MFO} implements \refalgobounds and \texttt{FMplex-c-MFO} implements \refalgobacktracking\!, both using the Min-Fanout heuristic.
Our approach is also compared to non-incremental implementations \texttt{FM} and \texttt{Simplex}.
The FMplex variants and \texttt{FM} always first employ Gaussian elimination to handle equalities.

All solvers were tested on the SMT-LIB \cite{BarFT-SMTLIB} benchmark set for QF\_LRA containing $1753$ formulas.
As all evaluated solvers are non-incremental, we also generated conjunctions of constraints by solving each of these QF\_LRA problems using a DPLL(T) SMT solver with an \texttt{FMplex-c-MFO} theory solver backend, and extracting all conjunctions passed to it.
If the solver terminated within the time and memory limits, we sampled 10 satisfiable and 10 unsatisfiable conjunctions (or gathered all produced conjunctions if there were fewer than 10).
This amounted to 3084 (777 sat, 2307 unsat) additional benchmarks.
The experiments were conducted on identical machines with two Intel Xeon Platinum 8160 CPUs (2.1 GHz, 24 cores).
For each formula, the time and memory were limited to 10 minutes and 5 GB.

\begin{table}[h]
    \center
    \scalebox{0.9}{
    \begin{tabular}{lrrrrr|rrrrr}%
            {} & \multicolumn{5}{c}{SMT-LIB} & \multicolumn{5}{c}{Conjunctions}\\
            {} &  solved &  sat &  unsat &  TO &  MO &  solved &  sat &  unsat &  TO &  MO\\
            \hline
            \texttt{Simplex}       &     958 &  527 &    431 &      714 &      81 &    3084 &  777 &   2307 &        0 &       0\\
            \texttt{FM}           &     860 &  461 &    399 &      577 &     316 &    2934 &  747 &   2187 &      107 &      43\\
            \texttt{FMplex-a-MFO}    &     814 &  432 &    382 &      840 &      99 &    2962 &  743 &   2219 &      122 &       0 \\
            \texttt{FMplex-a-MCL}    &     820 &  435 &    385 &      830 &     103 &    2965 &  742 &   2223 &      119 &       0 \\
            \texttt{FMplex-a-Rand-1} &     742 &  383 &    359 &      906 &     105 &    2806 &  668 &   2138 &      278 &       0 \\
            \texttt{FMplex-a-Rand-2} &     743 &  383 &    360 &      905 &     105 &    2823 &  671 &   2152 &      261 &       0 \\
            \texttt{FMplex-b-MFO} &     822 &  434 &    388 &      830 &     101 &    2988 &  744 &   2244 &       96 &       0 \\
            \texttt{FMplex-c-MFO} &     920 &  499 &    421 &      733 &     100 &    3084 &  777 &   2307 &        0 &       0 \\
            \texttt{Virtual-Best} &     982 &  532 &    450 &      651 &     120 &    3084 &  777 &   2307 &        0 &       0
        \end{tabular}
        }
        \caption{Number of solved instances, timeouts (TO) and memory-outs (MO).}\label{tbl:results}\vspace*{-0.5em}
    \end{table}

\begin{figure}[h!]
    \centering
    \begin{subfigure}{0.46\textwidth}
        \includegraphics{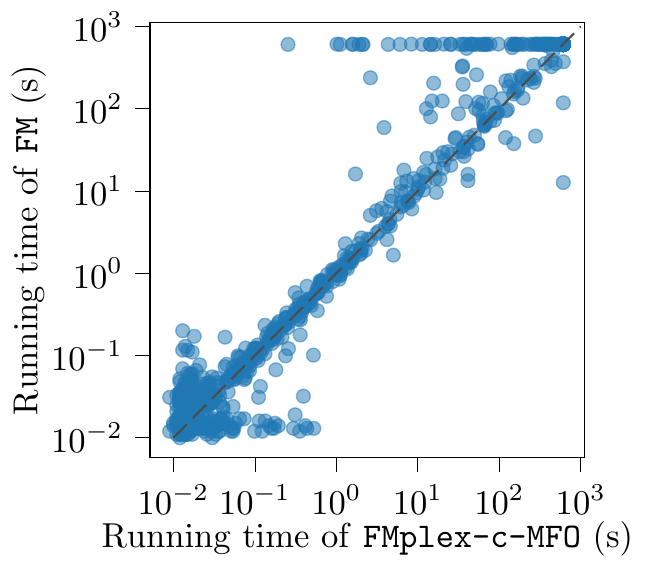}%
\caption{Running times in seconds on the SMT-LIB benchmark set, \texttt{FMplex} vs \texttt{FM}.}\label{fig:runtimes-a}
    \end{subfigure}%
    \hfill
    \begin{subfigure}{0.46\textwidth}
        \includegraphics{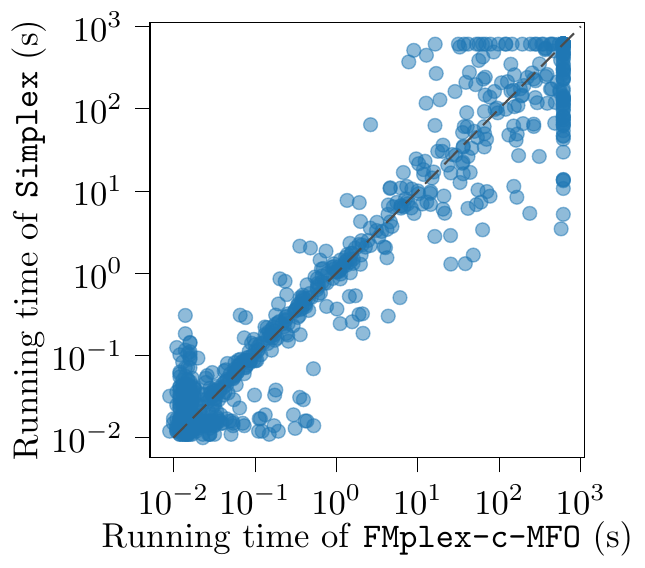}%
\caption{Running times in seconds on the SMT-LIB benchmark set, \texttt{FMplex} vs \texttt{Simplex}.}
        \label{fig:runtimes-b}
    \end{subfigure}\vspace*{4ex}
    \begin{subfigure}{0.46\textwidth}
        \includegraphics{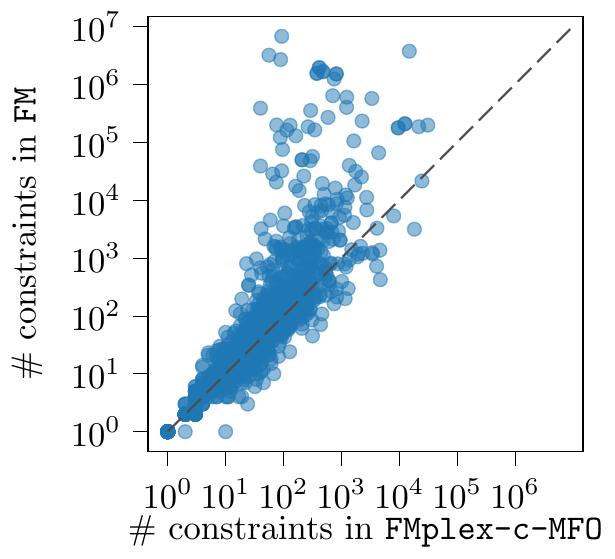}%
\caption{Number of generated constraints on the conjunctive benchmark set.}\label{fig:stats-conjunctions-a}
    \end{subfigure}%
    \hfill
    \begin{subfigure}{0.46\textwidth}
        \includegraphics{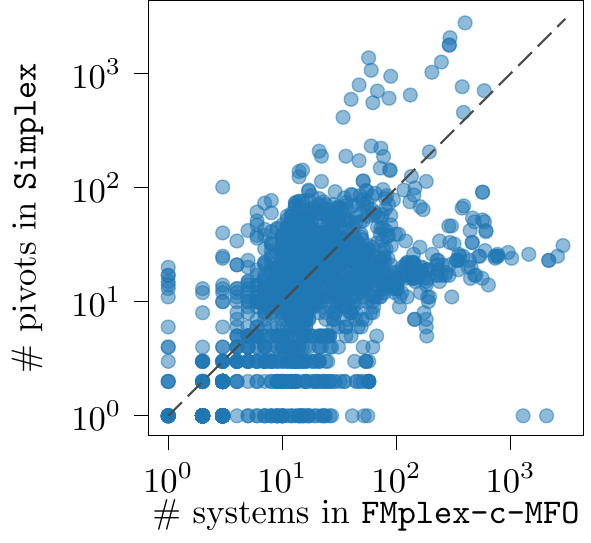}%

\caption{Number of visited non-bases (intermediate systems) on the conjunctive benchmark set.}\label{fig:stats-conjunctions-b}
    \end{subfigure}
\caption{%
Scatter plots: Each dot represents a single instance.
In (a) and (b), instances at the very top or right exceeded the resource limit.
Such instances are not considered in (c) and (d).%
}\label{fig:stats-conjunctions}

\vspace*{-2ex}
\end{figure}

The results in \Cref{tbl:results} show that \texttt{Simplex} solved the most SMT-LIB instances, followed by our \texttt{FMplex-c-MFO} and then \texttt{FM}.
Interestingly, \texttt{FM} solves fewer conjunctive instances than the base version of FMplex due to higher memory consumption ($43$ memory-outs for \texttt{FM}, while the others have none).
We see that a reasonable variable heuristic makes a difference as \texttt{FMplex-a-Rand-*} perform much worse than \texttt{FMplex-a-MFO} and \texttt{FMplex-a-MCL}.
However, between the latter two, there is no significant difference.
While our first optimization used in \texttt{FMplex-b-MFO} has no big impact, the backtracking implemented in \texttt{FMplex-c-MFO} allows for solving more instances within the given resource limits.

The running times for each individual SMT-LIB instance depicted in \Cref{fig:runtimes-a,fig:runtimes-b} reveal that \texttt{FM} and \texttt{FMplex-c-MFO} often behave similar, but \texttt{FM} fails on a number of larger instances.
We suspect that the smaller intermediate systems of FMplex are a main factor here.
While \texttt{Simplex} is often faster than \texttt{FMplex-c-MFO} and solves 61 SMT-LIB instances not solved by \texttt{FMplex-c-MFO}, it fails to solve 23 instances on which \texttt{FMplex-c-MFO} succeeds (Of these instances, \texttt{FM} solves 3 respectively 14 instances). Accordingly, the \texttt{Virtual-Best} of the tested solvers performs significantly better than just \texttt{Simplex}, indicating potential for a combination of \texttt{Simplex} and \texttt{FMplex-c-MFO}.

\Cref{fig:stats-conjunctions-a} compares the number of constraints generated by \texttt{FM} and \texttt{FMplex-c-MFO} on the conjunctive inputs.
Especially on larger instances, FMplex seems to be in the advantage.
Motivated by \Cref{subsec:avoiding-redundancies}, \Cref{fig:stats-conjunctions-b} compares the number of \texttt{Simplex} pivots  to the number of systems in \texttt{FMplex-c-MFO}.
We see that neither is consistently lower than the other, though \texttt{Simplex} seems to be slightly superior.
Due to the log-log scale, not shown are 1305 instances in which either measurement is 0 (920 instances for \texttt{Simplex}, 981 for \texttt{FMplex-c-MFO}).

\vspace*{0.25em}
\noindent The implementation and collected data are available at \url{https://doi.org/10.5281/zenodo.7755862}.

\section{Conclusion}
\label{sec:conclusion}

We introduced a novel method \emph{FMplex} for quantifier elimination and satisfiability checking for conjunctions of linear real arithmetic constraints.
Structural observations based on Farkas' Lemma and the Fundamental Theorem of Linear Programming allowed us to prune the elimination or the search tree.
Although the new method is rooted in the FM method, it has strong similarities with both the virtual substitution method and the simplex method.

The experimental results in the context of SMT solving show that FMplex is faster than Fourier-Motzkin and, although simplex is able to solve more instances than FMplex, there is a good amount of instances which can be solved by FMplex but cannot be solved by simplex.

In future work, we aim to combine the structural savings of FMplex with the efficient heuristic of simplex, i.e. we transfer ideas from FMplex to simplex and vice-versa.
Furthermore, we will investigate in tweaks and heuristics.
For instance, we plan to adapt the perfect elimination ordering from \cite{li2021choosing} and work on an incremental adaption for SMT solving.
Last but not least, we plan to increase the applicability of FMplex as a quantifier elimination procedure, including a different handling of strict inequalities, which is more similar to FM.

\bibliographystyle{eptcs}
\bibliography{literature.bib}

\ifextended

\appendix
\section{Proofs}
\input{appendix.tex}

\fi

\end{document}